\documentclass[10pt,conference]{IEEEtran}
\usepackage{amsmath,amsfonts,amsthm}
\usepackage{algorithmic}
\usepackage{algorithm}
\usepackage{array}
\usepackage{textcomp}
\usepackage{stfloats}
\usepackage{url}
\usepackage{color}
\usepackage{verbatim}
\usepackage{graphicx}
\usepackage{cite}
\usepackage{tikz}
\usepackage{tikzscale}
\usepackage{circuitikz}
\usepackage{subfigure}

\newtheorem{lem}{Lemma}

\DeclareMathOperator{\Si}{\mathrm{Si}}
\DeclareMathOperator{\Cin}{\mathrm{Cin}}

\usepackage{etoolbox}
\makeatletter
\patchcmd{\@makecaption}
  {\scshape}
  {}
  {}
  {}
\makeatother

\belowdisplayskip=\abovedisplayskip
\IEEEoverridecommandlockouts
\begin{document}

\title{On the Distribution of Matched Filtering with Continuous Aperture Arrays

\author{

\IEEEauthorblockN{Amy S. Inwood\IEEEauthorrefmark{1}, Abdulla Firag\IEEEauthorrefmark{3}, Peter J. Smith\IEEEauthorrefmark{2},~and Michail Matthaiou\IEEEauthorrefmark{1}}

\IEEEauthorblockA{\IEEEauthorrefmark{1}Centre for Wireless Innovation (CWI), Queen’s University Belfast, Belfast BT3 9DT, U.K.}

\IEEEauthorblockA{\IEEEauthorrefmark{2}School of Mathematics and Statistics, Victoria University of Wellington, Wellington, New Zealand}

\IEEEauthorblockA {Emails: \{a.inwood, m.matthaiou\}@qub.ac.uk, firag@yahoo.com, peter.smith@vuw.ac.nz}
    \vspace{-0.5cm}
}    

\thanks{This work was supported by the U.K. Engineering and Physical Sciences Research Council (EPSRC) grant (EP/X04047X/2) for TITAN Telecoms Hub. The work of P. J. Smith was supported by the Marsden Fund Council from New Zealand Government funding, managed by Royal Society Te Apārangi. The work of M. Matthaiou was supported by the European Research Council (ERC) under the European Union’s Horizon 2020 research and innovation programme (grant agreement No. 101001331).}
\thanks{\IEEEauthorrefmark{3}Abdulla Firag is a consultant with the Asian Development Bank. Prior to this, he was a director with Dhiraagu, a telecommunications company in the Maldives.}}

\maketitle

\begin{abstract}
Continuous aperture arrays (CAPAs) provide a theoretical upper bound on the performance of densely packed antenna arrays, but their analysis is limited by the lack of closed-form signal-to-noise ratio (SNR) distributions under realistic fading conditions. This paper derives accurate analytical expressions for the matched-filter SNR distribution of one-dimensional CAPAs in correlated Rayleigh environments under both the sinc and Jakes correlation models using the Karhunen–Loève expansion. By applying a truncated hypoexponential model, we obtain accurate approximations for the probability density function and cumulative distribution function of the SNR that closely match simulations, including the outage probability region where precise characterization is critical. Compared to a standard gamma approximation, our approach provides significantly improved accuracy in this regime. Additionally, the CAPA system considered is shown to outperform discrete antenna arrays. The derived expressions enable tractable and accurate evaluation of CAPAs under practical channel models.
\end{abstract}

\begin{IEEEkeywords}
Continuous aperture array, cumulative distribution function, matched filtering, probability density function.
\end{IEEEkeywords}
\vspace{-0.1em}
\section{Introduction}

The design of high-performance wireless communication systems increasingly depends on maximizing the spatial utilization of antenna apertures. Densely packing radiating elements within a fixed area enhances spatial degrees of freedom and enables improved beamforming and higher capacity. This has driven interest in antenna architectures that approach the theoretical limit of aperture densification.

Technologies in this space include continuous surface-based architectures, such as continuous reconfigurable intelligent surfaces (RISs) \cite{inwood_performance_2025}, holographic multiple-input-multiple-output (MIMO) systems \cite{demir_channel_2022} and stacked intelligent metasurfaces (SIMs) \cite{an_stacked_2024}. Another emerging technology is continuous aperture arrays (CAPAs) \cite{yuanwei_capa_2025, zhao_continuous_2025}, which operate as single antennas characterized by continuously distributed currents over their apertures. The concept of continuous aperture antennas traces back to foundational work in the 1960s \cite{wheeler_simple_1965}, where antennas were first idealized as smooth radiating surfaces. By modeling the antenna in this way, CAPAs provide valuable theoretical benchmarks and serve as upper bounds on the performance achievable by densely packed antenna arrays.

Accurate analytical characterization of CAPAs remains challenging because the distribution of the signal-to-noise ratio (SNR) is generally unknown. This undermines a rigorous performance analysis and makes it difficult to gain deep insights into system behavior. Thus, this work addresses this gap by developing accurate analytical expressions for the matched filtering output distribution for one-dimensional CAPAs in correlated Rayleigh fading environments. We approximate the continuous channel under both the sinc and Jakes correlation models using a Karhunen–Loève (KL) expansion \cite{loeve_probability_1978}, which enables the decomposition of the channel into uncorrelated components and thus facilitates a tractable analysis. The sinc model is considered directly, while the Jakes model is expressed via a ray-based channel (RBC) representation \cite{RBC} that converges exactly to 
Jakes correlation as the number of rays approaches infinity. In both cases, we truncate the infinite expansions to obtain computationally efficient and accurate approximations to the probability density function (PDF) and cumulative distribution function (CDF) of the SNR.

We investigate the accuracy of the derived SNR distributions, with particular attention to the lower tail critical for outage probability analysis. We also examine how the aperture length and carrier frequency impact system performance and how CAPAs perform against discrete antenna arrays, providing valuable insights for the design of future CAPAs and ultra-dense antenna systems.

\textit{Notation}: Statistical expectation is denoted $\mathbb{E}[\cdot]$; $\mathrm{Var}[\cdot]$ is the variance; $f_X(\cdot)$ and $F_X(\cdot)$ are the probability density function and cumulative distribution function, respectively, of a random variable $X$; $U[a,b]$ is a uniform distribution between minimum $a$ and maximum $b$; $\mathcal{CN}(\mu,\sigma^2)$ is a complex Gaussian distribution with mean $\mu$ and variance $\sigma^2$; $\mathrm{Exp}(\mu)$ is an exponential distribution with a mean of $\mu$; $\mathrm{Gamma}(\alpha,\delta)$ is a gamma distribution with shape parameter $\alpha$ and rate parameter $\delta$; $\Si(\cdot)$ and $\Cin(\cdot)$ are the sine and cosine integral functions, respectively; $\Gamma(\cdot)$ is the gamma function; $\gamma(\cdot,\cdot)$ is the lower incomplete gamma function; $(x*y)$ is the convolution of $x$ and $y$; $(\cdot)^*$ is the complex conjugate; $||\cdot||$ is the Euclidean norm and i.i.d.  denotes independent and identically distributed.

\section{System Model}\label{sec:sysmodel}
Consider the uplink one-dimensional (1D) CAPA system shown in Fig.~\ref{fig:sys_model}, where all points along the aperture of length $W$ can be activated. A single-user (SU) system is analyzed in this work to build up foundational theory on CAPAs under matched filtering, and the more complex scenarios of multi-user and zero-forcing systems are left for future work. The channel from a single-antenna source to position 
$x$ on the aperture is denoted by $h(x)$, with $0\leq x \leq W$.

\begin{figure}[h]
\centering
\resizebox{0.3\textwidth}{!}{\includegraphics{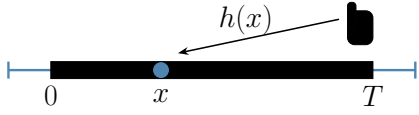}}
\vspace{-1.5em}
\caption{The UL system model for a 1D continuous aperture antenna.}
\vspace{0.5em}
\label{fig:sys_model}
\end{figure}

In this work, we consider correlated Rayleigh fading channels under both the sinc and Jakes correlation models. For the sinc correlation case, we directly model the channel as
\begin{equation}
    h(x) \sim \mathcal{CN}(0,\beta),
\end{equation}
where $\beta$ is the channel gain, while the covariance between channels at positions $x_1$ and $x_2$ is \cite{bjornson_rayleigh_2021}:
\begin{equation}
    \mathbb{E}\left[h(x_1)h^*(x_2)\right]=C(x_1,x_2) = \mathrm{sinc}\left(\frac{2\pi}{\lambda}(x_1-x_2)\right). \label{eq:sinc_corr}
\end{equation}

For analytical tractability, we employ a ray-based representation of correlated Rayleigh channels exhibiting Jakes correlation \cite{RBC}, where
\begin{equation}
    h(x)= \sum_{r=1}^R\sqrt{\beta_r}e^{j\Phi_r}\,\mathrm{exp}\left(\frac{2\pi j}{\lambda}  x \sin(\theta_r)\right),
\end{equation}
and $R$ is the number of received rays, $\lambda$ is the wavelength, $\theta_r$ is the azimuthal angle-of-arrival (AoA) of ray $r$, $\beta_r$ is the channel gain of ray $r$, and $\sum_{r=1}^R\beta_r=\beta$. The phase offsets of each ray, $\Phi_r$ are assumed to be i.i.d. $U[0,2 \pi]$. The covariance function, averaged over the random phase offsets, is \cite{RBC}
\begin{equation}
    C(x_1,x_2) = \sum_{r=1}^R\beta_r\,\mathrm{exp}\left(\frac{2\pi j}{\lambda} \sin(\theta_r)(x_1-x_2)\right). \label{eq:raycorr}
\end{equation}
As $R \rightarrow \infty $ the discrete covariance converges to the Jakes isotropic scattering model, $C(x_1,x_2) \rightarrow J_0(2\pi(x_1-x_2)/\lambda)$, when rays are uniformly distributed over $[0, 2\pi]$ and have equal power ($\beta_r=\beta/R$). For both correlation models, we assume that $\beta$ is constant across all antenna positions. 

The received signal at $x$ for both correlation models is
\begin{equation}
    r(x) = h(x)s + n,
\end{equation}
where $s$ is the transmitted signal, with $\mathbb{E}[|s|^2]=E_s$ and $n\sim\mathcal{CN}(0,\sigma^2)$ is the additive white Gaussian noise. Assuming matched filtering of the received signal at every point on the CAPA, the resulting SNR is 
\begin{equation}
    \mathrm{SNR}=\frac{E_s}{\sigma^2}\int_0^W|h(x)|^2dx=\frac{E_s}{\sigma^2}\gamma_1.
\end{equation}
Accurately predicting system performance metrics, such as the outage probability and the bit error rate, requires knowledge of the distribution of $\gamma_1$. This distribution is not available in closed form.  However, it can be analytically characterized using the KL expansion, derived in the following section.

\section{Analysis}
In this section, the KL expansion is employed to represent the channel in a form that facilitates the derivation of analytical expressions for the PDF and CDF of the SNR.

\subsection{KL Expansion of $\gamma_1$}
\label{sec:KLexp}
The KL expansion decomposes a Gaussian process into a sum of orthogonal eigenfunctions of its covariance operator, which must be Hermitian, multiplied by the corresponding eigenvalues and independent Gaussian random variables. This exact representation fully characterizes the distribution of the process. Therefore, $h(x)$ can be exactly expressed as \cite{loeve_probability_1978}\footnote{The KL expansion  used in this work can be generalized to 2D CAPAs. However, the covariance operator would be defined over a two-dimensional domain and the associated eigenvalues would involve quadruple integrals, making closed-form expressions more challenging to obtain. Hence, we focused on the 1D scenario in this initial work.}
\begin{equation}
    h(x)=\sum_{n=1}^\infty\sqrt{\lambda_n}u_n(x)z_n, \label{eq:KL_form}
\end{equation}
where $\lambda_n$ is the $n$-th eigenvalue, $u_n$ is the $n$-th eigenfunction of $C(x_1,x_2)$, while $z_n\sim\mathcal{CN}(0,1)$ are i.i.d. complex Gaussian variables. Substituting \eqref{eq:KL_form} into $\gamma_1$ gives
\begin{align}
    \gamma_1 &= \int_0^W\left|\sum_{n=1}^\infty\sqrt{\lambda_n}u_n(x)z_n\right|^2dx,\notag \\
    &= \sum_{n=1}^\infty\sum_{m=1}^\infty\sqrt{\lambda_n\lambda_m}z_n^*z_m\int_0^W\!\!\!u^*_n(x)u_m(x)dx.
\end{align}
As the covariance operator is Hermitian, the eigenfunctions $u_n(x)$ and $u_m(x)$ can be selected to be orthonormal such that
\begin{equation}
    \int_0^W\!\!\!u^*_n(x)u_m(x)dx= \begin{cases} 
      1 & n=m, \\   
      0 & n\neq m, 
   \end{cases}
\end{equation}
which leads to
\begin{equation}
    \gamma_1=\sum_{n=1}^\infty\lambda_n|z_n|^2, \label{eq:KLgamma}
\end{equation}
where $|z_n|^2\sim\mathrm{Exp(1)}$. The independence of the random variables in \eqref{eq:KLgamma} yields the following concise expressions for the mean and variance of $\gamma_1$:
\begin{equation}
    \mathbb{E}[\gamma_1]=\sum_{n=1}^\infty\lambda_n\,\mathbb{E}[|z_n|^2]=\sum_{n=1}^\infty\lambda_n = \frac{E_s \beta W}{\sigma^2}, \label{eq:Egamma1}
\end{equation}
\begin{equation}
    \mathrm{Var}[\gamma_1]= \sum_{n=1}^\infty\lambda_n^2\,\mathrm{Var}\!\left[|z_n|^2\right] = \sum_{n=1}^\infty\lambda_n^2. \label{eq:Vargamma1}
\end{equation}

\subsection{Eigenvalues of the Covariance Function} 
\label{sec:eigs}
To derive the eigenvalues of $C(x_1,x_2)$ for both the sinc and Jakes correlation models, the following Fredholm integral equation of the second kind must be solved \cite{loeve_probability_1978}: 
\begin{equation}
    \int_0^W\!\!C(x_1,x_2)u_n(x_2)dx_2=\lambda_n u_n(x_1), \,\,\, x_1\in[0,W].
\end{equation}
Due to the orthonormality of the eigenfunctions, multiplying both sides by $u_n^*(x_1)$ and integrating over $x_1$ gives
\begin{equation}
    \lambda_n = \int_0^W\int_0^W\!\!C(x_1,x_2)\,u_n^*(x_1)\,u_n(x_2)\,dx_1dx_2. \label{eq:eigvals}
\end{equation}
In both cases, we use cosine basis functions to approximate the eigenfunctions of the covariance operator, such that 
\begin{equation}
    u_n(x) = \sqrt{\frac{\epsilon_n}{W}}\cos\left(\frac{n\pi x}{W}\right), \label{eq:cos_basis}
\end{equation}
where $\epsilon_0=1$ and $\epsilon_n=2$ when $n>0$. Over the entire real line, both the sinc and Jakes correlation functions are translation invariant, making their exact eigenfunctions complex exponentials \cite{rasmussen_gaussian_2006}. As these kernels are real and even, the eigenbases can equivalently be taken as cosine functions over the full domain.  However, translation invariance is lost on finite intervals, and the eigenfunctions differ. Bandlimited kernels, such as sinc, have eigenfunctions that lack closed-form expressions \cite{papoulis_probability_2002}, while non-bandlimited kernels, like Jakes, require numerical computation or approximations. In practice, cosines form a complete orthonormal basis on finite intervals and provide an accurate approximation, making them a convenient and effective choice.

%For the ray-based scenario, if $R\rightarrow\infty$ and the rays approach a uniform distribution over $[0, 2\pi]$ with equal power,  the discrete ray-based covariance converges to the classical Jakes isotropic scattering model, $C(x_1,x_2) \rightarrow J_0(2\pi(x_1-x_2)/\lambda)$, which has cosine functions as its exact eigenbasis \cite{rasmussen_gaussian_2006}. This makes the KL expansion exact in this case. 

%In the correlated Rayleigh scenario the cosine basis is not the exact eigenbasis for the sinc covariance kernel. The true eigenfunctions are prolate spheroidal wave functions (PSWFs) \cite{papoulis_probability_2002}, which account for both finite spatial support and limited spatial bandwidth. However, since PSWFs lack closed-form expressions and are computationally involved, we approximate them using cosines. Cosines form a complete orthonormal basis on the interval and closely match the dominant modes, offering a practical and effective alternative.

\subsubsection{Sinc Correlation}
Using \eqref{eq:sinc_corr} for $C(x_1,x_2)$ and \eqref{eq:cos_basis} for $u_n(x_1)$ and $u_n(x_2)$, \eqref{eq:eigvals} becomes
\begin{equation}
    \lambda_n\!\approx\!\frac{\epsilon_n}{W}\!\!\!\int_0^W\!\!\!\int_0^W\!\!\!\!\mathrm{sinc}\!\left(\!c(x_1{-}x_2)\!\right)\!\cos\!\left(a_nx_1\!\right)\!\cos\!\left(a_nx_2\!\right)\!dx_1dx_2,\!\! \label{eq:lambdan_rayleigh}
\end{equation}
with $a_n=\frac{n\pi}{W}$ and $c=\frac{2\pi}{\lambda}$. The solution to \eqref{eq:lambdan_rayleigh} is derived in Lemma \ref{lem:rayleigh}.

\begin{lem}
\label{lem:rayleigh}
    The $n$-th eigenvalue of $C(x_1,x_2)$ in \eqref{eq:sinc_corr} can be accurately approximated by \eqref{eq:lambda_sinc}, where $A_2=2a_nW$, $C_2=2cW$, $C_p{=}W(c{+}a_n)$ and $C_m{=}W(c{-}a_n)$.
    
    \begin{figure*}[t]
    \normalsize
    \begin{equation}
         \label{eq:lambda_sinc}
         \lambda_n\!\approx\!\! \begin{cases}\!          \frac{2W}{c}\Si(cW) + \frac{2}{c^2}(\cos(cW)-1)   & a_n = 0, \\
          \!\frac{2W}{C_2}\Si(C_2)\left(1+\frac{\sin(C_2)}{2c}\right) + \frac{2W}{C_2^2}\left(\cos(C_2)(\Cin(C_2)+1)+\Cin(C_2)-1\right)   & a_n=c, \\
          \!\frac{2}{C_2}\!\!\left[\!\frac{W}{A_2}\!(\!(\sin(A_2)\!+\!A_2)\!(\Si(C_p)\!+\!\Si(C_m)\!)\!+\!(1\!+\!\cos(A_2)\!)\!(\Cin(C_p)\!-\!\Cin(C_m)\!)\!) \!+\!\frac{\cos(C_p)-1}{c+a_n}\!+\!\frac{\cos(C_m)-1}{c-a_n}\!\right]\!\!\! & \mathrm{otherwise}. 
         \end{cases}\!
    \end{equation}
    \hrulefill
    \vspace{-1.25em}
    \end{figure*}
\end{lem}
\begin{proof}
    See Appendix A.
\end{proof}
\subsubsection{Jakes Correlation}
Using \eqref{eq:raycorr} for $C(x_1,x_2)$ and \eqref{eq:cos_basis} for $u_n(x_1)$ and $u_n(x_2)$, \eqref{eq:eigvals} becomes
\begin{equation}
    \lambda_n\!\!\approx\!\!\frac{\epsilon_n\beta}{WR}\!\sum_{r=1}^R\!\int_0^W\!\!\!\!\!\mathrm{e}^{j\alpha_rx_1}\!\cos\!\left(a_nx_1\!\right)\!\!\!\int_0^W\!\!\!\!\!\mathrm{e}^{-j\alpha_rx_2}\!\cos\!\left(a_nx_2\!\right)\!dx_1dx_2, \label{eq:lambdan_ray}
\end{equation}
where $\alpha_r=c\sin(\theta_r)$. The solution to \eqref{eq:lambdan_ray} is derived in Lemma \ref{lem:ray}.

\begin{lem}
    \label{lem:ray}
    When $R \to \infty$ and \eqref{eq:raycorr} converges to the Jakes model, the $n$-th eigenvalue of $C(x_1, x_2)$ in \eqref{eq:raycorr} can be accurately approximated by
    \begin{equation}
        \lambda_n\!\approx\!\frac{\beta}{2WR}\!\sum_{r=1}^R\!\begin{cases}
            2W^2 & a_n{=}0, |\alpha_r|{=}a_n, \\
            2\left|\frac{\mathrm{exp}(j\alpha_rW)-1}{j\alpha_r}\right|^2 & a_n{=}0, |\alpha_r|{\neq} a_n, \\
            W^2 & a_n{\neq} 0,|\alpha_r|{=} a_n, \\
            \begin{aligned}&\bigg|\tfrac{\mathrm{exp}(j(\alpha_r+a_n)W)-1}{j(\alpha_r+a_n)} \\ & \quad+\tfrac{\mathrm{exp}(j(\alpha_r-a_n)W)-1}{j(\alpha_r-a_n)}\bigg|^2\end{aligned}\!\!\!& a_n {\neq} 0, |\alpha_r| {\neq} a_n.
        \end{cases} \label{eq:lambda_ray}
    \end{equation}
\end{lem}
\begin{proof}
    See Appendix B.
\end{proof}

\subsection{PDF and CDF of $\gamma_1$}
\label{sec:PDFCDF}
From \eqref{eq:KLgamma}, $\gamma_1$ can be expressed as a generalized chi-squared random variable. However, closed-form probability distributions for this class of variables are generally unavailable. Instead, by truncating the infinite sum to the \(N\) dominant terms, $\gamma_1$ can be approximated by a hypoexponential random variable, which has a known distribution. In this work, we consider this case, and an enhanced model where a gamma correction term is added to account for the discarded smaller eigenvalues. The latter approach is particularly useful when many eigenvalues are similar in magnitude, which can otherwise lead to instability in the hypoexponential approximation. 

For strictly bandlimited kernels, such as sinc correlation, Landau's theorem provides a useful estimate of the number of dominant modes \cite{chong_electromagnetic_2025}: as the aperture length $W$ increases, roughly $\frac{2W}{\lambda}$ eigenvalues approach unity while the remainder are near zero. Even for finite $W$, a sharp decay occurs in the eigenvalues after the first 
$\frac{2W}{\lambda}$, ensuring that the estimate remains reliable. For non-bandlimited kernels, such as Jakes correlation, the number of large eigenvalues remains approximately $\frac{2W}{\lambda}$, but the decay is more gradual, so additional eigenvalues contribute non-negligibly to the channel energy.

\subsubsection{Hypoexponential Distribution}
Truncating \eqref{eq:KLgamma} to $N$ terms, we can write $\gamma_1 \approx \tilde{\gamma}_1{=} \sum_{n=1}^NX_n$, where $ X_n\!\sim\mathrm{Exp}({1}/{\lambda_n})$. Thus, using \cite[Eq. (5.8)]{ross_introduction_1997}, the PDF of $\tilde{\gamma}_1$ is
\begin{equation}
    f_{\tilde{\gamma}_1}(x)=\sum_{n=1}^N\frac{\lambda_n^{N-2}\mathrm{e}^{-x/\lambda_n}}{\prod_{m\neq n}(\lambda_n-\lambda_m)}. \label{eq:PDFtrunc}
\end{equation}
Integrating \eqref{eq:PDFtrunc} gives the CDF of $\tilde{\gamma}_1$ as
\begin{equation}
    F_{\tilde{\gamma}_1}(x) = 1-\sum_{n=1}^N\frac{\lambda_n^{N-1}\mathrm{e}^{-x/\lambda_n}}{\prod_{m\neq n}(\lambda_n-\lambda_m)}. \label{eq:CDFtrunc}
\end{equation}
\subsubsection{Gamma Variable Correction} 
Here, we again truncate $\gamma_1$ to $N$ terms, but add a gamma variable as a further correction factor to improve accuracy. We approximate $\gamma_1$ by $\tilde{\gamma}_g=  \tilde{\gamma}_1 +Y$, where $Y\sim\mathrm{Gamma}(r,\theta)$ approximates the discarded terms, $\sum_{n=N+1}^{\infty}\!X_n$. We fit $Y$ by the method of moments, solving $ 
\mathbb{E}[\gamma_1]=\mathbb{E}[\tilde{\gamma}_1]+\mathbb{E}[Y]$ and $
\mathrm{Var}[\gamma_1]=\mathrm{Var}[\tilde{\gamma}_1]+\mathrm{Var}[Y]$ for the moments of $Y$, noting that the exact moments of $\gamma_1$  are given in \cite{erfan} and the moments of $\tilde{\gamma}_1$ can be found by truncating \eqref{eq:Egamma1} and \eqref{eq:Vargamma1}. We obtain $r$ and $\theta$ from $ {r}/{\theta}{=}\mathbb{E}[Y]$ and ${r}/{\theta^2}{=}\mathrm{Var}[Y]$. As $Y$ is a gamma variable,
\begin{equation}
    f_Y(y)=\frac{\theta^ry^{r-1}\mathrm{e}^{-\theta y}}{\Gamma(r)}.
\end{equation}
As $\tilde{\gamma}_1$ and $Y$ are independent, the PDF of $\tilde{\gamma}_g$ is therefore
\begin{align}
    f_{\tilde{\gamma}_g}(x) \!&\approx (f_{\tilde{\gamma}_1}\,*\,f_{Y})(x)=\int_0^xf_{\tilde{\gamma}_1}(x-y)f_Y(y)dy, \notag \\
    &\!\approx\!\sum_{n=1}^N\!\frac{\theta^r\lambda_n^{N-2}\mathrm{e}^{-x/\lambda_n}}{\Gamma(r)\prod\limits_{m\neq n}(\lambda_n\!-\!\lambda_m)}\int_0^xy^{r-1}\mathrm{e}^{-y(\theta-1/\lambda_n)}dy, \notag \\
    &\!\approx\!\sum_{n=1}^N\!\frac{\theta^r\lambda_n^{N-2}\mathrm{e}^{-x/\lambda_n}}{\Gamma(r)(\theta\!-\!1/\lambda_n\!)^r\!\!\!\prod\limits_{m\neq n}\!\!(\lambda_n\!\!-\!\lambda_m)\!}\gamma\Big(\!r,\!\Big(\!\theta\!-\!\tfrac{1}{\lambda_n}\!\!\Big)x\!\Big)\!.\! \label{eq:PDFgamma}
\end{align}
The approximate CDF of $\tilde{\gamma}_g$ can be expressed as a convolution involving \eqref{eq:CDFtrunc} and \eqref{eq:PDFgamma}, such that
\begin{align}
    F_{\tilde{\gamma}_g}\!(x)\!&\approx (F_{\tilde{\gamma}_1}\,*\,f_{\tilde{\gamma}_g})(x)=\int_0^xF_{\tilde{\gamma}_1}(x-y)f_{\tilde{\gamma}_g}(y)dy,\notag\\ &\!\approx\!\!\frac{\gamma(r,\theta x)}{\Gamma(r)}\!-\!\!\!\sum_{n=1}^N\!\!\frac{\theta^r\lambda_n^{N-1}\mathrm{e}^{-x/\lambda_n}}{\Gamma(r)\!(\theta\!-\!\!1\!/\!\lambda_n\!)^r\!\!\!\!\prod\limits_{m\neq n}\!\!\!(\!\lambda_n\!\!-\!\!\lambda_m\!)\!}\gamma\!\Big(\!r,\!\!\Big(\!\theta\!-\!\tfrac{1}{\lambda_n}\!\!\Big)x\!\Big)\!. \label{eq:CDFgamma}
\end{align}

\section{Numerical Results}\label{sec:numresults}
This section verifies the analytical results and explores the system behavior. In Figs. \ref{fig:CDFs}-\ref{fig:OP}, we consider the PDF, CDF and outage probability of ${\gamma}_1$, respectively, which in turn validates the KL expansion and the corresponding PDFs and CDFs. Where not otherwise specified, $W=1$ and $f=800$ MHz. For all results, the channel gain $\beta$ is selected so that the mean SNR when $W=1$ and $f=800$ MHz is 0 dB. We assume that $N=100$, which in turn corresponds to $100$ eigenvalues. This value was chosen to exceed the dominant mode estimate of $\frac{2W}{\lambda}$ discussed in Sec.~\ref{sec:PDFCDF} for all parameter sets considered, ensuring that the majority of the channel energy is captured. This large number causes many tail eigenvalues to have similar values, so the gamma-corrected PDF and CDF in \eqref{eq:PDFgamma} and \eqref{eq:CDFgamma} are used. For the Jakes correlation results, we assume $R=200$. For all simulations, $10^7$ replicates were generated.

\subsection{CDF Validation and Impact of Antenna Aperture}
Figure \ref{fig:CDFs} validates the accuracy of the KL expansion of $\gamma_1$ in Secs. \ref{sec:KLexp} and \ref{sec:eigs} and the accuracy of the approximate CDF in \eqref{eq:CDFgamma}, as well as comparing CAPAs with discrete antenna arrays of the same length. For a fair comparison, each discrete antenna is seen as a segment of the 1D line, and the corresponding discrete channel coefficient is the aggregate of the channel over the segment. We consider two discrete cases with 8 equally spaced elements within $W$, where the element apertures and the inter-element spacing are chosen such that the array captures 50\% and 80\% of the energy of the continuous aperture, respectively.

\begin{figure}[ht]
  \centering
  \subfigure[Sinc correlation]{
    \includegraphics[width=0.45\linewidth]{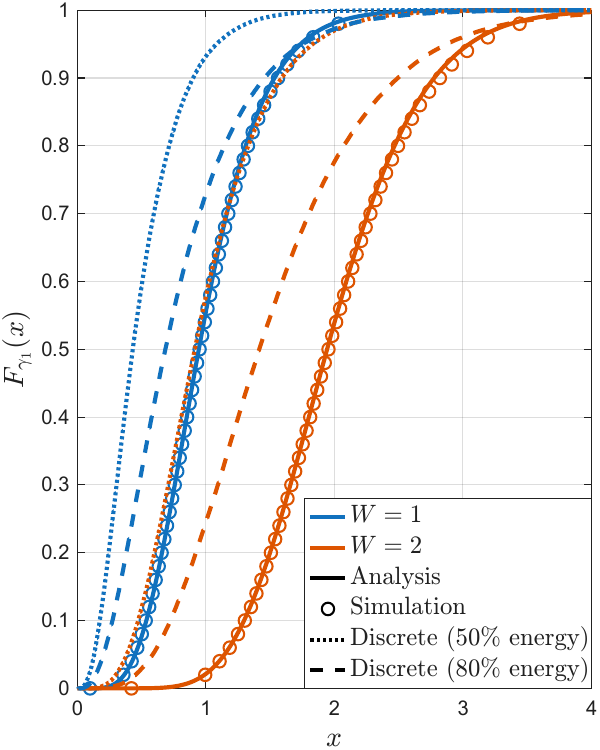}
    \label{fig:CDFsinc}
  }
  \hfill
  \subfigure[Jakes correlation]{
    \includegraphics[width=0.45\linewidth]{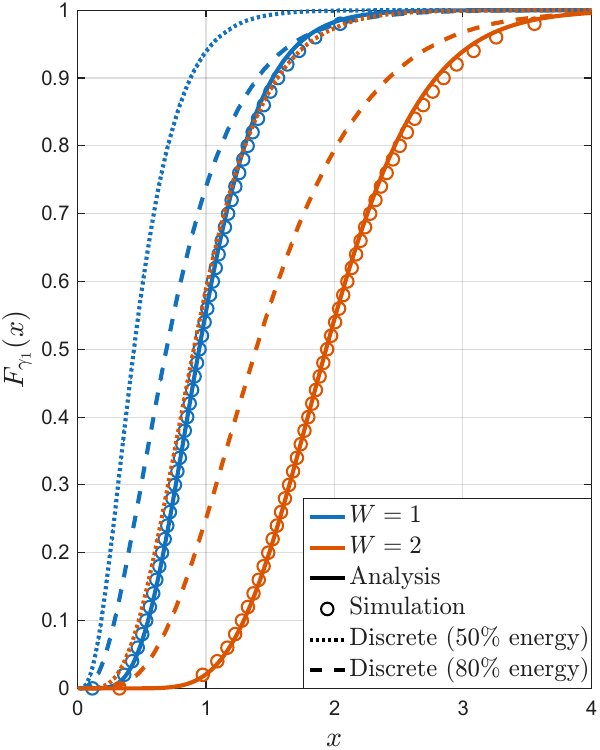}
    \label{fig:CDFray}
  }
  \caption{CDFs of CAPAs of varying lengths for both correlation models compared with discrete antenna arrays of the same length.}
  \label{fig:CDFs}
\end{figure}

The agreement between the analytical and simulated CAPA results is excellent, showing that the derived expressions accurately capture the SNR distribution of CAPAs. Additionally, it can be seen that CAPAs outperform discrete arrays of the same length. While discrete arrays have lower correlation due to antenna spacing, CAPAs exploit the full aperture to capture more energy, yielding superior performance and motivating their adoption in next-generation mobile systems.

Figure \ref{fig:CDFs} shows that the median SNR of CAPA systems grows nearly linearly with the aperture length $W$, reflecting the linear scaling of the mean SNR in \eqref{eq:Egamma1}. For $W=1$, the analytical mean and observed median are 1.00 and 0.95, respectively, while for $W=2$, they are 2.00 and 1.94. The close agreement of these values indicates that the SNR distributions are only weakly skewed. This approximate symmetry arises from the decay of spatial correlation and the averaging effect of integrating over many partially correlated contributions. Thus, the distribution becomes increasingly Gaussian-like for larger $W$, making the mean a good indicator of the median.

\subsection{PDF Validation and Impact of Carrier Frequency}
Figure \ref{fig:sinc_pdf} compares analytical and simulated CAPA system PDFs. Only results for sinc correlation are shown to avoid overcrowding the plot. Three aperture lengths and two carrier frequencies are considered.

\begin{figure}[ht]
    \centering
    \includegraphics[scale=0.56]{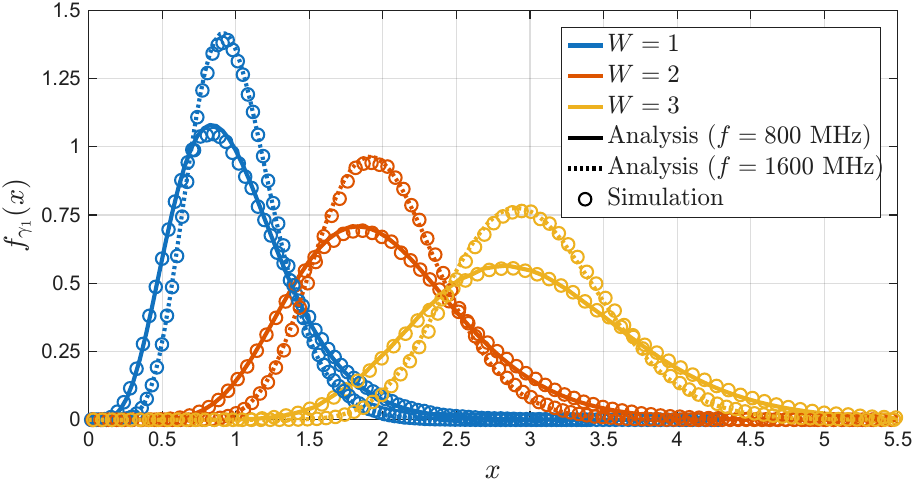}
    \vspace{-2em}
    \caption{The PDFs of CAPAs for correlated Rayleigh channels.}
    \label{fig:sinc_pdf}
\end{figure}

Once again, there is very good agreement between the analytical results and simulations across a range of antenna aperture lengths and carrier frequencies. As the aperture length $W$ increases, the SNR increases proportionally, resulting in a rightward shift of the PDFs. 

Increasing both $W$ and spatial correlation leads to wider distributions, indicating higher variance. This is caused by changes in the eigenvalue spectrum of the spatial correlation kernel directly affecting the variance of the SNR, as seen in \eqref{eq:Vargamma1}.
As $W$ increases, the number of significant eigenvalues increases, each positively contributing to the variance. 

The spatial correlation increases when the carrier frequency decreases. This increases the wavelength and results in a slower decay of correlation with distance. High spatial correlation leads to a concentration of power into a smaller number of dominant eigenvalues. This concentration results in a larger variance, as the sum of the squared eigenvalues increases.

Variance quantifies the absolute spread of the SNR distribution but ignores changes in the average signal strength. The coefficient of variation (CV), $\mathrm{CV}=\sqrt{\mathrm{Var}[\gamma_1]}/\mathbb{E}[\gamma_1]$, provides a normalized measure of variability by indicating the magnitude of fluctuations relative to the average SNR. CV values for each scenario in Fig. \ref{fig:sinc_pdf} are listed in Tab. \ref{tab:CV}.

\begin{table}[ht]
\centering
\caption{CV of the Scenarios in Fig. \ref{fig:sinc_pdf}}
\vspace{-1em}
\begin{tabular}{|c|c|c|c|}
\hline
    & $W=1$ & $W=2$ & $W=3$ \\ \hline
$f=800$ MHz &  0.40  & 0.29  & 0.24  \\ \hline
$f=1600$ MHz &  0.29 &  0.21 & 0.17  \\ \hline
\end{tabular}
\label{tab:CV}
\end{table}
\vspace{0.2em}

Although increasing both $W$ and spatial correlation leads to higher variance, they have opposite effects on the CV. As $W$ grows, the mean increases faster than the standard deviation due to the addition of more eigenmodes, causing the CV to decrease and indicating improved relative stability. In contrast, higher spatial correlation concentrates power into fewer modes, increasing the standard deviation more than the mean, which raises the CV and reflects greater relative fluctuations. Therefore, while both factors increase absolute variability, only correlation degrades relative reliability.

\subsection{Outage Probabilities}
Figure \ref{fig:OP} investigates the outage probability of ${\gamma}_1$ under sinc correlation. The probability of ${\gamma}_1$ being below a threshold $\gamma_{\mathrm{th}}$ is accurately approximated by
\begin{equation}
    P_{\text{out}} = P({\gamma}_1 \leq \gamma_{\text{th}}) \approx P(\tilde{\gamma}_g \leq \gamma_{\text{th}}) = F_{\tilde{\gamma}_g}(\gamma_{\text{th}}).
\end{equation}
Outage probability is an important metric to study in this work for multiple reasons. CAPAs mitigate outage probability by averaging small-scale multipath fading across their surface, thereby reducing deep fades that lead to performance degradation. 
The outage probability is also an important test of the analysis presented in this work, as it focuses on the region where analytical inaccuracies are most likely to manifest. To benchmark the accuracy of the derived distribution, we compare it with a gamma approximation, commonly used to approximate the distribution of the instantaneous SNR in fading channels. To emphasize the region where outage probability is most critical, we plot these results on a log scale. For these results, $f=400$ MHz and $W<1$, modeling scenarios with a high coefficient of variation. These parameters represent challenging environments where the average SNR is a poor predictor of reliability, and an accurate distribution is particularly important.

\begin{figure}[ht]
    \centering
    \includegraphics[scale=0.56]{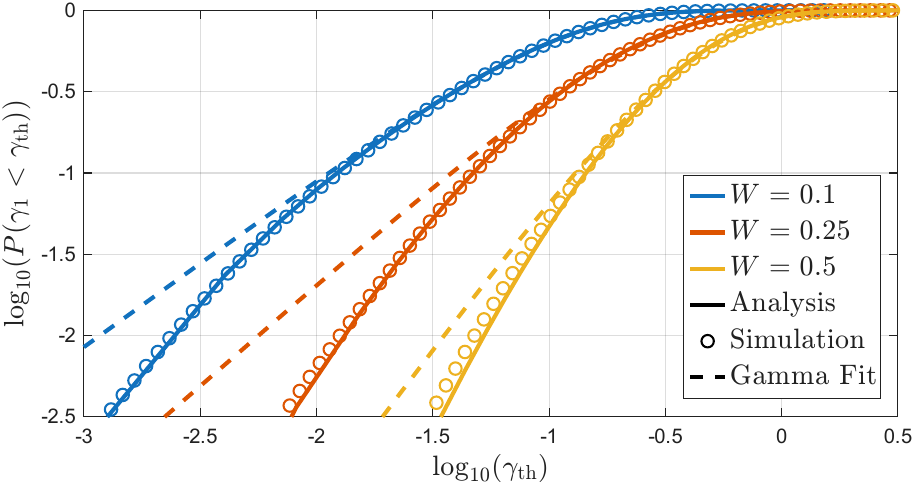}
    \vspace{-2em}
    \caption{The outage probabilities of CAPAs under sinc correlation.}
    \label{fig:OP}
\end{figure}

Figure \ref{fig:OP} shows that the distribution derived in this work accurately characterizes the outage probability. The analytical expression closely matches the simulated results across the entire range of thresholds, with excellent agreement even at very low outage probabilities of the order of $10^{-2.5}$. In contrast, the gamma approximation exhibits noticeable discrepancies in the high-reliability region, substantially overestimating the outage probability. This is particularly pronounced for smaller aperture sizes. This highlights the value of the proposed KL-based distributions. Additionally, the outage probability curves shift to the right as the aperture size increases, reflecting the combined effects of greater energy capture and spatial averaging. These effects increase the average  SNR and reduce the likelihood of deep fades.

%\begin{figure}[ht]
  %\centering
  %\subfigure[Rayleigh channels.]{
    %\includegraphics[width=0.45\linewidth]{sinc_OP.pdf}
    %\label{fig:OPsinc}
  %}
  %\hfill
  %\subfigure[Ray-based channels.]{
    %\includegraphics[width=0.45\linewidth]{ray_OP.pdf}
    %\label{fig:OPray}
 % }
  %\caption{Outage probabilities of CAPAs of varying lengths under both correlated Rayleigh and ray-based channels. FIX THIS REMOVE THE 10}
  %\label{fig:OPs}
%\end{figure}

\section{Conclusion}
This paper developed highly accurate analytical approximations for the matched-filter SNR distribution of 1D CAPAs in correlated Rayleigh environments under both sinc and Jakes correlation using a KL expansion. The analysis demonstrated excellent agreement with simulations and outperformed a discrete antenna array. Compared to a standard gamma approximation, our approach provided improved modeling of outage behavior, which is critical for reliability analysis. These results offer a practical framework for evaluating the performance of CAPAs and support their use as a benchmark for ultra-dense antenna arrays.

\section*{Appendix A \\ Derivation of Sinc Kernel Eigenvalues}
Rewriting the sinc function, we get
\begin{align}
    \lambda_n\!\!\approx\!\tilde{\lambda}_n\!\!&=\!\frac{\epsilon_n}{W}\!\!\!\int_0^W\!\!\!\!\!\int_0^W\!\!\!\frac{\sin(c(x_1{-}x_2)\!)\!}{c(x_1-x_2)}\!\cos\!\left(a_nx_1\!\right)\!\cos\!\left(a_nx_2\!\right)\!dx_{\!1}dx_2, \notag \\ 
    &=\frac{\epsilon_n}{W}I_n. \notag
\end{align}
A change of variable is required to avoid the difference on the denominator, resulting in a division by 0. Therefore, let $t=x_1-x_2\in[-W,W]$. The change in the integration limits partitions the original integral into two separate integrals, together covering the full integration range:
\begin{multline}
    I_n= \int_{-W}^0\!\int_{-t}^W\frac{\sin(ct)}{ct}\cos(a_n(t\!+\!x_2))\cos(a_nx_2)dx_2dt \\ + \int_{0}^W\!\!\!\!\int^{W\!-t}_0\!\frac{\sin(ct)}{ct}\cos(a_n(t\!+\!x_2))\cos(a_nx_2)dx_2dt.\!\!
\end{multline}
We consider three scenarios: the general case $a_n\neq c \neq 0$, and the special cases $a_n=c$ and $a_n=0$. Since all cases follow the same derivation approach, we present the detailed steps for the general case only. The special cases can be obtained analogously and are omitted here for brevity. 

Using a sum-to-product identity and integrating over $x_2$,
\begin{align}
    I_n&{=}\frac{1}{4a_nc}\bigg[\int_{-W}^0\!\!\frac{\sin(ct)}{t}\!\cos(a_nt)(\sin(A_2){+}A_2{+}\sin(2a_nt)  \notag \\ & \quad {+}2a_nt)dt {+}\!\!\int_{-W}^0\!\!\!\frac{\sin(ct)}{t}\sin(a_nt)\!\left(\cos(A_2){-}\cos(2a_nt)\!\right)\!dt \notag \\
    &\quad+ \int_0^W\frac{\sin(ct)}{t}\cos(a_nt)(\sin(A_2-2a_nt){+}A_2{-}2a_nt)dt \notag \\ &\quad+\int_0^W\frac{\sin(ct)}{t}\sin(a_nt)\left(\cos(A_2-2a_nt)-1\right)dt\bigg] \notag, \\
    &=\!\frac{1}{4a_nc}[I_{n1}+I_{n2}+I_{n3}+I_{n4}], \label{eq:Isum}
\end{align}
where $A_2=2a_nW$. Expanding $I_{n1}$,
\begin{multline}
    I_{n1} \!=\!2a_n\!\!\int_{\!-\!W}^0\!\sin(ct)\!\cos(a_nt)dt+\!\!\int_{\!-\!W}^0\!\!\frac{\cos(a_nt)\!\sin(2a_nt)}{t} \\ \times \sin(ct)dt+(\sin(A_2)+A_2)\int_{-W}^0\frac{\sin(ct)\cos(a_nt)}{t}dt. \notag
\end{multline}
Applying product-to-sum identities to separate the trigonometric integrals, and knowing that the sine snd cosine integrals are defined as $\Si(x)=\int_0^x\frac{\sin(t)}{t}$ and $\Cin(x)=\int_0^x\frac{\cos(t)-1}{t}$,
\begin{align}
&I_{n1} {=} a_n\!\!\left[\!\tfrac{\cos(C_p){-}1}{c+a_n}{+}\tfrac{\cos(C_m){-}1}{c-a_n}\!\right]\!\!{+}\tfrac{1}{4}\!(\Cin(C_{3p}\!){-}\Cin(C_{3m}\!)  \notag \\ &\,\,\,{+}\Cin(C_p){-}\Cin(C_m)){+}\tfrac{1}{2}\!(\sin(A_2){+}A_2)(\Si(C_p){+}\Si(C_m)\!),\! \label{eq:I1}
\end{align}
where $C_p=W(c+a_n)$, $C_p=W(c-a_n)$, $C_{3p}=W(c+3a_n)$ and $C_{3m}=W(c-3a_n)$. Following the same procedure,
\begin{align}
    I_{n2} &= \tfrac{1}{2}\cos(A_2)(\Cin(C_p){-}\Cin(C_m))+\tfrac{1}{4}(\Cin(C_p) \notag \\ &-\Cin(C_m)-\Cin(C_{3p})+\Cin(C_{3m})).\label{eq:I2} \\
    I_{n3} &= a_n\!\left[\!\tfrac{\cos(C_p){-}1}{c+a_n}{+}\tfrac{\cos(C_m){-}1}{c-a_n}\!\right]\!{+}\tfrac{1}{2}A_2(\Si(C_p){+}\Si(C_m)\!) \notag \\ &-\! \tfrac{1}{4}\!\cos(A_2)\!(\Cin(C_{3m}){-}\!\Cin(C_{3p}){+}\!\Cin(C_m){-}\!\Cin(C_p)\!) \notag \\ & +\!\tfrac{1}{4}\!\sin(A_2)\!(\Si(C_{3p}){+}\Si(C_{3m}){+}\Si(C_p){+}\Si(C_m)\!).\!\label{eq:I3} \\
    I_{n4} &= \tfrac{1}{2}(\Cin(C_p){-}\Cin(C_m)){+}\tfrac{1}{4}\cos(A_2)(\Cin(C_{3m})\notag \\&{-}\!\Cin(C_{3p}){-}\Cin(C_m){+}\Cin(C_p)){+}\tfrac{1}{4}\!\sin(A_2) \notag \\ &{\times} (\Si(C_p)+\Si(C_m)-\Si(C_{3p})-\Si(C_{3m})).\label{eq:I4}
\end{align}
Substituting \eqref{eq:I1}-\eqref{eq:I4} into \eqref{eq:Isum} and simplifying gives the general case in \eqref{eq:lambda_sinc}. Following the same process with $a_n=0$ and $a_n=c$ respectively gives the two special cases in \eqref{eq:lambda_sinc}.

\section*{Appendix B \\ Derivation of Jakes Kernel Eigenvalues }
Using the identity $\cos(a_nx)=(\mathrm{e}^{ja_nx}+\mathrm{e}^{-ja_nx})/2$, \eqref{eq:lambdan_ray} becomes
\begin{multline}
    \lambda_n\!\approx\!\tilde{\lambda}_n\!=\!\frac{\epsilon_n\beta}{4WR}\!\sum_{r=1}^R\!\int_0^W\!\!\!\!\left(\mathrm{e}^{j(\alpha_r+a_n)x_1}{+}\mathrm{e}^{j(\alpha_r-a_n)x_1}\right)dx_1 \\ \times\int_0^W\!\!\!\!\left(\mathrm{e}^{-j(\alpha_r-a_n)x_2}{+}\mathrm{e}^{-j(\alpha_r+a_n)x_2}\right)dx_2. 
\end{multline}
As the two integrals are complex conjugates with identical limits, their product can be expressed as the squared magnitude of a single integral, such that
\begin{equation}
    \tilde{\lambda}_n\!\!=\!\frac{\epsilon_n\beta}{4WR}\!\!\sum_{r=1}^R\!\left|\int_0^W\!\!\!\!\!\left(\!\mathrm{e}^{j(\alpha_r+a_n)x+j(\alpha_r-a_n)x}\!\!\right)\!\!dx\right|^2\!\!\!\!=\!\frac{\epsilon_n\beta}{4WR}\!\!\sum_{r=1}^R\!|I_r|^2\!\!. \label{eq:rayint}
\end{equation}
The integral $I_r$ is now evaluated for four configurations of parameters. Case 1 involves $a_n=0$, $|\alpha_r|=a_n$, such that
\begin{equation}
    I_{r,c1} = \int_0^W\!\!2\,dx = 2W. \label{eq:rayC1}
\end{equation}
Case 2 considers $a_n=0$, $|\alpha_r|\neq a_n$, giving 
\begin{equation}
    I_{r,c2} = \int_0^W\!\!\!2\mathrm{e}^{j\alpha_rx}dx=\left[\frac{2\mathrm{e}^{j\alpha_rx}}{j\alpha_r}\right]^W_0\!\!\!\!=2\left(\!\frac{\mathrm{e}^{j\alpha_rW}-1}{j\alpha_r}\!\right)\!. \label{eq:rayC2}
\end{equation}
Case 3 involves $a_n\neq0$, $|\alpha_r|=a_n$, and as $\alpha_r$ is an integer multiple of $\frac{\pi}{W}$, we obtain
\begin{equation}
    I_{r,c3}=\int_0^W\!\!\!\!\left(\mathrm{e}^{2j\alpha_rx}+1\right)dx=\left[\frac{\mathrm{e}^{2j\alpha_rx}}{2j\alpha_r}{+}x\right]^W_0=W. \label{eq:rayC3}
\end{equation}
Finally, case 4 considers $a_n\neq0$, $|\alpha_r|\neq a_n$, leading to
\begin{align}
    I_{r,c4} &= \int_0^W\!\!\!\!\mathrm{e}^{j(\alpha_r+a_n)x}dx+\int_0^W\!\!\!\!\mathrm{e}^{j(\alpha_r-a_n)x}dx, \notag \\
    & = \tfrac{\mathrm{e}^{(j(\alpha_r+a_n)W)}-1}{j(\alpha_r+a_n)} +\tfrac{\mathrm{e}^{(j(\alpha_r-a_n)W)}-1}{j(\alpha_r-a_n)}. \label{eq:rayC4}
\end{align}

\noindent Substituting \eqref{eq:rayC1}-\eqref{eq:rayC4} into \eqref{eq:rayint} gives \eqref{eq:lambda_ray}.

\bibliographystyle{IEEEtran}
\bibliography{IEEEabrv, referencesCTS.bib}

\end{document}